\documentclass[smallextened]{svjour3}
\usepackage{amsmath,amssymb}

\usepackage{mathrsfs} 
\usepackage{graphicx}	
\usepackage{multirow} 

\usepackage{lineno}			
\usepackage{enumerate} 	
\usepackage{xfrac}					

\smartqed

\newcommand{\pos}[1]{\left[#1\right]^{+}}										

\newcommand{\mb}[1]{\ensuremath{\mathbf{#1}}}
\newcommand{\wh}[1]{\widehat{#1}}														
\newcommand{\wt}[1]{\widetilde{#1}}													

\spnewtheorem{assumption}{Assumption}{\bf}{\it}

\newcommand{\squishlist}{
   \begin{list}{$\bullet$}
    { \setlength{\itemsep}{0pt}      \setlength{\parsep}{1.5pt}
      \setlength{\topsep}{3pt}       \setlength{\partopsep}{0pt}
      \setlength{\leftmargin}{1.5em} \setlength{\labelwidth}{1em}
      \setlength{\labelsep}{0.5em} } }

\newcommand{\squishend}{
    \end{list}  }

\begin{document}
\title{Fixed and Market Pricing for Cloud Services
}
\author{Vineet Abhishek \thanks{This work was done while Vineet Abhishek was interning at Microsoft Research}\and Ian Kash \and Peter Key}
\institute{Vineet Abhishek
 \at Walmart Labs,  USA \\ \email{vineet.abhishek@gmail.com}
 \and Ian Kash
  \at Microsoft Research, 21 Station Road, Cambridge, CB1 2FB, UK,  \\ \email{iankash@microsoft.com}
  \and  Peter Key
    \at Microsoft Research, 21 Station Road, Cambridge, CB1 2FB, UK,  \\      \email{peter.key@microsoft.com}
  }
%
\date{\today}

\maketitle
%
%

\begin{abstract}
We study a model of congestible resources, where pricing and scheduling are intertwined.  Motivated by the problem of pricing cloud instances, we model a cloud computing service
as linked $GI/GI/\cdot$ queuing systems  where the provider chooses to offer a fixed pricing service, a dynamic market based service, or a hybrid of both, where jobs can be preempted in the market-based service. Users (jobs), who are heterogeneous in  both the value they place on service and  their cost for waiting,  then choose between the services offered.
Combining insights from auction theory with queuing theory we are able to characterize user equilibrium behavior,  and show its insensitivity to the precise market design mechanism used.  We then provide theoretical and simulation based evidence suggesting that a fixed price typically, though not always,  generates a higher expected revenue than the hybrid system for the provider.
\end{abstract}

\section{Introduction} \label{sec:introduction}
\textit{Cloud computing} provides on-demand and scalable access to computing resources. Public clouds, such as Windows Azure and Amazon EC2, treat infrastructure computing as a service (IaaS) that can be purchased and delivered over the Internet.  An agent purchases units of computing time on virtual machines (referred to as \textit{instances}). The most commonly used pricing mechanism for instances is \textit{pay as you go} (henceforth, PAYG), where an agent is charged a fixed price per unit time per instance. However, given stochastic demand, such fixed pricing may result in unused resources. Rather than letting resources sit idle, the provider could operate a \textit{spot market}, selling unused resources at a reduced price via an auction  to agents willing to tolerate delays and interruptions; indeed, Amazon EC2 runs such a market for Spot Instances.  Ben-Yehuda et al.~\cite{AgmonBen-Yehuda:2013} attempt 
a retrospective deconstruction of how the Spot Instances are priced.

This paper examines the trade-offs for a provider deciding whether or not to operate a spot market. On one hand, operating a spot market can create price discrimination, as agents with low valuations and low waiting costs compete for spot instances, thereby extracting payments from the agents who would balk if PAYG were the only option. On the other hand, the spot market provides a cheaper alternative for agents with high valuations but low waiting cost, causing a loss of revenue from the agents who would have paid a higher PAYG price if PAYG were the only option.   In consequence,  it is not obvious if operating PAYG and the spot market simultaneously provides any net gain in the expected revenue to the cloud service provider.  Nor is it clear how operating a spot market affects social welfare: while adding a spot market causes some new agents to be served, some existing agents that were already being served will now join the spot market, increasing their waiting costs.  Furthermore, if the cloud services provider desires to extract a given amount of revenue, while maximizing welfare subject to that constraint, he would set different prices in each system.

To quantify the trade-offs  we construct a simple, representative  model of a cloud computing service with agents who are heterogeneous both in their value for service and in their waiting cost. We first analyze PAYG and a spot market in isolation and use the resulting insights to analyze what happens when they operate simultaneously. Our analysis is not tied to any particular pricing rule for the spot market. Instead, we use a characterization similar to the revenue equivalence theorem for auctions \cite{Myerson81} to characterize the expected payment made by an agent in any equilibrium of any pricing rule. Moreover, while the analysis of the queuing system with multiple priority classes and multiple servers is complex (see, e.g., \cite{Harchol-Balter-etal2005}, \cite{Heijden-etal2004}) an application of the revelation principle \cite{Myerson81} allows us to circumvent this complexity. We  describe a general queuing system for the spot market purely in terms of a waiting time function and exploit its properties for our analysis.    Throughout our paper, by  ``waiting time''' we mean the waiting time in the \emph{system} or sojourn time.  

In summary,  we  make four  contributions in this paper:

\begin{enumerate}[(i)]
\item
We combine insights from auction theory with queuing theory: we model a cloud computing service as a queuing system described by a waiting time function, and then apply techniques from the theory of optimal auctions to analyze it.  

\item
We characterize agent behavior, and  show that, in the unique equilibrium, agents have a waiting cost threshold that determines whether they participate in the spot market or PAYG. Moreover, their bids in the spot market are independent of their value for service and increasing in their waiting cost\footnote{Throughout this paper, ``increasing'' means ``strictly increasing.''}.

\item
Using this equilibrium characterization, we provide theoretical and simulation evidence suggesting that operating PAYG in isolation generally, though not always, provides a higher expected revenue to the cloud service provider than operating PAYG and a spot market simultaneously.  This is under the assumption that there is no other cost to running a job in the spot market, e.g.  that there is not cost in being preempted.

\item
We prove that, taking the PAYG price as a real cost (not just a transfer), agents make \emph{efficient}  decisions about whether to join the spot market; we  give  simulation evidence describing the tradeoffs between efficiency in revenue under each system.
\end{enumerate}

While our results are based on a stylized model of a cloud computing system, we also discuss how the assumptions of our model can be relaxed and the implications of our results for the decision  Amazon has made to run a spot market.

\subsection{Related Work}
Our work is at the nexus of queuing theory and game theory. Hassin and Haviv~\cite{Hassin&Haviv2002} provide a survey of this area. For observable $M/M/1$ queues with identical customers, Balachandran~\cite{Balachandran1972} derives a full information equilibrium. Hassin~\cite{Hassin1995} and Lui~\cite{Lui1985} consider unobservable $M/M/1$ queues where customers with heterogeneous waiting costs bid for preemptive priority using the first price auction. They characterize an equilibrium where bids are increasing in the waiting cost. Af\`{e}che and Mendelson~\cite{Afeche&Mendelson2004} extend this to more general waiting costs.
Dube and Jain~\cite{Dube&Jain2009} consider a different problem with competing $GI/GI/1$ priority queues; arriving jobs decide which queue to join. They find conditions for the existence of a Nash equilibrium.

  Closest to our work are papers that apply the theory of optimal auction design to optimize pricing and service policies in queuing system. Af\`{e}che~\cite{Afeche2004} and Af\`{e}che and Pavlin~\cite{afeche2013} show that delaying jobs or choosing orderings that increase processing time can increase revenue. Yahalom et al.~\cite{Yahalom-etal2006} generalize \cite{Afeche2004} by relaxing the distributional assumptions on valuation and working with convex delay cost. Katta and Sethuraman~\cite{Katta&Sethuraman2005} design a pricing scheme that, under some assumptions, is optimal for an M/M/1 queuing system and certain generalizations of it. Cui et al.~\cite{Cui-etal} consider the problem of jointly managing pricing, scheduling, and admission control policy for revenue maximization for $M/M/1$ queues and find solutions for some special cases.  Xu and Li~\cite{Xu2013} examine possibilities to improve revenue in a PAYG market through resource throttling.  Doroudi et al.\cite{Doroudi-et-al2013} study pricing in an $M/G/1$ queue for a single class, and a continuum of customers in that class.  They assume that waiting costs are proportional to valuations, which allows customer types to be unidimensional and hence the theory of optimal (revenue maximizing) auctions that flowed from Myerson's seminal work~\cite{Myerson81} can be applied. They derive closed form solution for the price functions when the customer values are drawn from specific types, and compare priority pricing with fixed pricing.  While their results show something similar in spirit to spot pricing raising more revenue than PAYG, this is driven by an assumption that their PAYG market is oversubscribed while this is not the case for us.

One issue we do not address is competition among cloud providers, an issue studied by Anselmi et al.~\cite{anselmi2014}.

 Compared to previous work in this literature, the distinguishing aspects of our work are: (i) we allow for an arbitrary queuing system with multiple servers and arrival process which need not be memoryless; (ii) our analysis is not tied to a specific auction mechanism for the spot market; (iii) we allow PAYG and the spot market to operate simultaneously and are not limited to analyzing a system in isolation; and (iv) we examine the tradeoff between efficiency and revenue.

\section{Model} \label{sec:model}
Consider a cloud computing system where jobs arrive sequentially according to a stationary stochastic process with independent interarrival times. Each job demands one instance and is associated with a distinct agent;  we will use the terms ``agents'' and ``jobs'' interchangeably.   The system designs the pricing and scheduling mechanism,   with the aim of maximizing revenue, while the jobs aim to maximize their expected payoff.

The service time for each job is independently drawn according to an arbitrary distribution with the expected time of $1/\mu$, where we  assume that the exact service time is unknown to anyone, including the job itself. Jobs differ in their values for service and their waiting costs. There are $n$ classes of jobs. Each job from class $i$ has the same value $v_i$ for job completion, and we assume $v_i > v_{i+1}$. The total arrival rate of potential jobs is $\lambda = \sum_i \lambda_i$. Each job is independently assigned class $i$ with probability $\lambda_i/\lambda$, hence the total arrival rate of potential jobs from class~$i$ is $\lambda_i$. Each job from class $i$ incurs a waiting cost per unit time which is an i.i.d. realization of a random variable $C_i$ with the cumulative distribution function (cdf) $F_i(c)$,  where $f_i(c)$ is  the corresponding probability density function (pdf) of $F_i(c).$  The class and exact waiting cost of a job is its private information; however, the class values $v_i$ and probability distributions $F_i$ are common knowledge. The random variable $C_i$'s are independent of each other.

Jobs (agents)  choose whether or not to enter the system,  are \emph{Individually Rational} and  risk neutral   with respect to payments  and benefits and hence aim to maximize  their expected payoff.  If a job from class $i$ with waiting cost $c$ pays an amount $m$ for using the instance and spends the total time $w$ in the system (the sum of the queuing time and the service time, referred to as the \textit{waiting time}), then the  full price to the job is the sum of the direct payment and indirect waiting costs, $cw+m$, and   then hence its payoff is $v_i - cw - m$.
Jobs compete for system resources by   submitting  a ``bid'' to the system, where the information contained in the bid will depend on the  mechanism design, and is some function of the job's value, $v_i$ and waiting cost $c_i$.

 Under our Individually Rational assumption, each job  competes to acquire an instance only if its expected payoff is nonnegative.  Hence, $f_i(c)$ is assumed to be strictly positive\footnote{This is not a restrictive assumption in practice since we can approximate any distribution arbitrarily well with such a distribution.  See, e.g., Figures \ref{fig:pdfs} and \ref{fig:beta-rev} and associated discussion} for $c \in [0, \mu v_i]$, (since jobs from class~$i$ with waiting cost greater than $\mu v_i$ will always balk).
%

Two pricing and scheduling regimes are allowed,  shown schematically in Figure \ref{fig:payg-spot-q}  and which we now describe.
\begin{figure} 
\begin{center}
\includegraphics[trim=0.0in 3.00in 5.00in 0.25in, clip=true,height=3.0in]{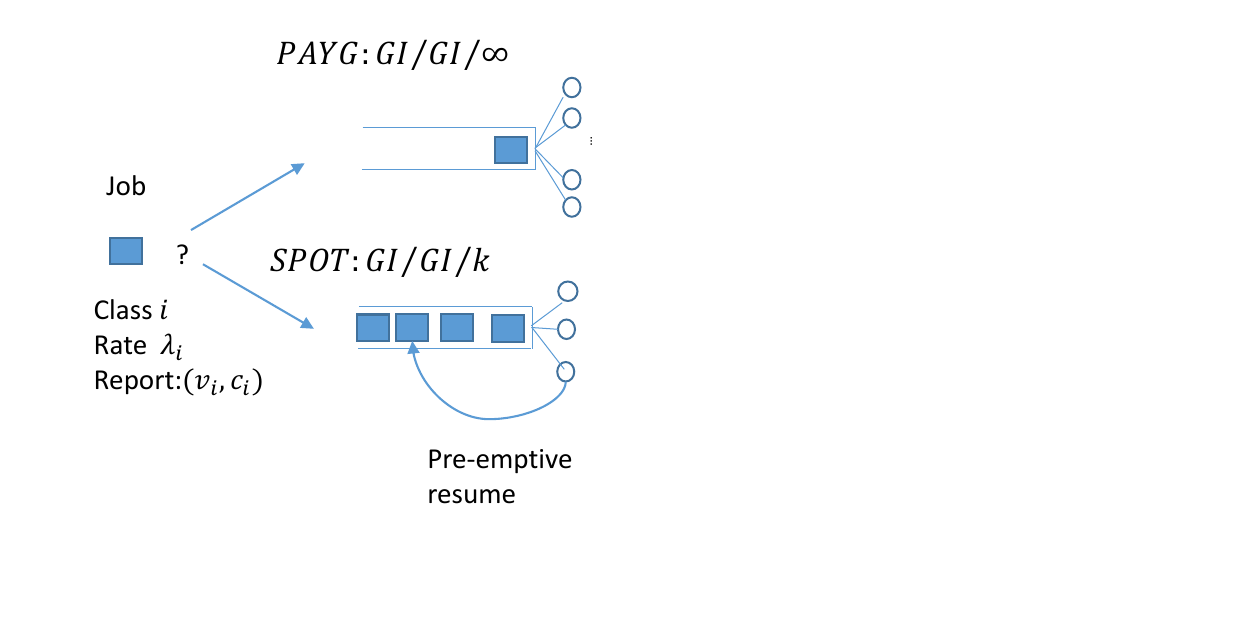}
\caption{System model for PAYG and spot market}.\label{fig:payg-spot-q}
\end{center}
\end{figure}

\textbf{Modeling PAYG}: We assume that the overall system has enough capacity to  serve the exogenous demand $\lambda / \mu$ with negligible buffering or rejection of individual jobs, and hence    PAYG is modeled as a $GI/GI/\infty$ system with service rate $\mu$  (this is discussed further in Section \ref{sec:finiteCap}). A job arriving to PAYG joins immediately and is served until completion. Each job is charged a price $p > 0$ per unit time for using a PAYG instance. The price $p$ is common knowledge. The only information contained in a ``bid'' from a job is a request to enter.   The expected payoff of a job from class $i$ with the waiting cost $c$ from using a PAYG instance is thus $v_i - (c+p)/\mu$. If $c > \mu v_i - p$, the job does not participate in PAYG. 

\textbf{Modeling the spot market}: The spot market is modeled as a $GI/GI/k$ system with preemption where jobs bid for priority, where a bid can be thought of (but does not have to be) a single number that represents a willingness to pay.    We will be working with auctions where a job with a higher bid is given priority over a job with a lower bid and can preempt the lowest priority job under service if needed; Section~\ref{sec:strategy} provides further details on the assumptions we make on the relationship between bids and priorities. A job which is preempted goes back to the queue and waits to resume from the point it left. The queue state
is unobservable to the arriving jobs. Jobs are not allowed to renege or change their bids. A job is charged based its own bid and the bids of others according to some spot pricing mechanism.
Examples include the first price auction where at each time the jobs with $k$ highest bids currently in the system
are served and each pays its bid, and the $(k+1)^{th}$ price auction where the jobs with $k$ highest bids are served and each job pays the current $(k+1)^{th}$ highest bid per-unit time. We do not explicitly assume any specific spot pricing mechanism and abstract away from it by considering the expected payment by a job in a Bayes Nash Equilibrium (henceforth, BNE) using the revenue equivalence theorem for auctions~\cite{Myerson81}. 
A consequence is that only the expected payment matters, and hence any payment implementation that achieves the correct payments can be used. In practice, charging via a  per-unit time price is a natural implementation.

\section{Mechanism Design} \label{sec:mechd}
We now recap certain fundamental concepts  from mechanism design, which we make use of in our analysis, and which may be unfamiliar to the queueing systems community.

Mechanism design is a branch of game theory that describes rules governing the allowed interactions of agents.  Given some allowed set of \emph{bids}
or reports from the buyers, which will depend on their underlying ``true'' reports or values, a mechanism prescribes both  an allocation rule (who should get what) and a payment rule,  how much they should pay, where the allocation and payment rules are functions of the submitted bids.
A standard notion of a stable solution relevant to this setting is that of a Bayes Nash Equilibrium (henceforth BNE).  In a BNE there is a prior distribution over the values of other agents and the strategy adopted by each agent is optimal in expectation given this prior, his value, and the strategies other agents of each type would adopt.

 The \emph{revelation principle}, states that if there is a BNE for some mechanism, then an equivalent direct  \emph{truthful} mechanism exists, one where it is an equilibrium for the  the buyers to bid their true values,  equivalent in the sense of outcomes (allocation and payments).
This result, proved by Myerson~\cite{Myerson79}, makes it sufficient to focus on truthful, direct mechanisms.  
 
 How to choose amongst competing (truthful direct) mechanisms?  In a seminal work,  Myerson~\cite{Myerson81} showed that for mechanisms that are truthful or \emph{incentive compatible} in the BNE sense, then first, the allocation rule leads to one and only one payment rule, and secondly that the payoff formula bounds the payoff of any feasible mechanism.  The first result is a \emph{Revenue Equivalence} result:  different mechanisms which lead to the same allocation rule provide the same revenue; the second allows the design of optimal mechanisms, where revenue is maximized.  In this paper, we use these techniques to characterize the equilibrium outcome once a PAYG price has been fixed.

One notable truthful mechanism is the Vickrey-Clarke-Groves or VCG mechanism,  which has the property that it is \emph{efficient}, i.e. maximizes social welfare (see, for example \cite{Krishna10} for a description of VCG).    It also has the property that the payment rule for an agent depends on the bids of other agents and represents the externality that agent imposes by being part of the system. For example, if $k$ identical goods are being sold, then the VCG mechanism is the Vickrey auction~\cite{Vickrey61}  that allocates the $k$ items to the highest $k$ bidders, and an agent wining pays the $k+1$st highest bid (i.e. the first losing bid).  In this paper, we make use of this characterization of the VCG mechanism when analyzing the social welfare created by a hybrid market in Section~\ref{sec:welfare}.


\section{PAYG and Spot Market Analysis} \label{sec:PAYG-Spot}

\subsection{Strategy, waiting time, and spot pricing} \label{sec:strategy}
When a spot market is operating, either alone or in conjunction with PAYG, a job that decides to join it participates in an auction and must decide how to bid based on the payment rules of the auction.  The optimal bid may depend in a complicated way on its private information (value for service and cost of waiting).  As a result, previous work has typically focused on analyzing a particular mechanism such as a first price auction.  However, we show in this section that this complexity is dispensable.  Regardless of the auction mechanism, jobs that enter the spot market with higher waiting costs pay more and wait less time and these values are (essentially) independent of the job's class.  The job's class does matter in determining whether a job participates in the spot market, but this takes the form of a simple cutoff: jobs with waiting costs below the cutoff participate and those above do not.

By the revelation principle for BNE \cite{Myerson79}, it suffices to restrict our consideration to truthful direct revelation mechanisms: mechanisms where jobs report their private information and it is an equilibrium for them to do so truthfully.  Any implementable outcome is implementable by such a mechanism.  Thus, a job reports a type $(v,c)$ and if it participates in the spot market has an expected waiting time $\wt{w}(v,c)$ and expected payment $\wt{m}(v,c)$.  In principle, these could depend on the value $v$ of the job's class, however, we show that it is essentially without loss to assume they do not. 

\begin{lemma} \label{lemma:class-independent}
For all truthful direct revelation mechanisms for the spot market and all equilibria there exists an equilibrium with the same expected utility where expected waiting time and payments are independent of class for all values of~$c$ where multiple classes participate in the spot market.
\end{lemma}
\begin{proof}
A job of class $i$ with waiting cost $c$ that participates in the spot market chooses a report $(v',c')$ minimizing the expected total cost
$c \wt{w}(v',c') + \wt{m}(v',c')$. Thus, when multiple classes participate, the set of optimal reports is class-independent; in particular, if classes $i$ and $j$ participate,  both $(v_i,c)$ and $(v_j,c)$ belong to the set of optimal reports.
Let $s_i(c)$ be the (randomized) equilibrium strategy for class $i$ with cost $c$. Now, suppose instead that every job from a class that participates with positive probability with waiting cost $c$ uses strategy $s_i(c)$ with probability $\lambda_i f_i(c) / (\sum_{j \in SPOT(c)} \lambda_j f_j(c))$ where $SPOT(c)$ is the set of all classes that participate in the spot market with positive probability with cost $c$. Then the arrival process for all strategies $s_i(c)$ remains identical to the original process.
That is, the probability that the next job to arrive reports a pair $(v,c)$ is unchanged.
Hence the waiting time and the expected payment remain unchanged. This new class-independent randomized strategy is also an equilibrium for all classes.
\qed \end{proof}

Since jobs can undo any tie-breaking a class-based mechanism does, we assume for the remainder of the paper that mechanisms have a class-independent expected waiting time $\wt{w}(c)$ and expected payment $\wt{m}(c)$.    This allows us to concentrate on one-dimensional reports (bids) from the agents, related to their waiting cost.
We now show that jobs with higher waiting costs pay more and spend less time waiting.

\begin{lemma} \label{lemma:w-m-monotonicity}
In (the truthful) equilibrium, $\wt{w}(c)$ is nonincreasing in $c$ and $\wt{m}(c)$ is nondecreasing in $c$ for values of $c$ that participate in the spot market for some class.
\end{lemma}
\begin{proof}
Consider $\wh{c} > c$. The optimality of truthful reporting implies:
\begin{align}
\wh{c}\wt{w}(\wh{c}) + \wt{m}(\wh{c}) & \leq \wh{c}\wt{w}(c) + \wt{m}(c), \label{eq:w-m-eq1} \\
c\wt{w}(c) + \wt{m}(c) & \leq c\wt{w}(\wh{c}) + \wt{m}(\wh{c}). \label{eq:w-m-eq2}
\end{align}
Adding \eqref{eq:w-m-eq1} and \eqref{eq:w-m-eq2} implies $\wt{w}(\wh{c}) \leq \wt{w}(c)$.
Using this and \eqref{eq:w-m-eq2}, we get $\wt{m}(\wh{c}) \geq \wt{m}(c)$.
\qed \end{proof}

Thus far, our assumptions have been without loss of generality.  We now make two assumptions that are not.
\begin{assumption}
  We assume that jobs with no waiting cost are served for free in the spot market, hence $\wt{m}(0) = 0$. 
 \end{assumption}
We revisit this assumption at the end of the paper.
 \begin{assumption}
We assume that, in equilibrium in the spot market, jobs with higher waiting costs always have strictly higher priority than jobs with lower waiting costs.
\end{assumption} Note that this assumption is a stronger condition than assuming that $\wt{w}(c)$ is decreasing. Since $\wt{w}$ is the expected waiting time, if priorities are assigned randomly it is possible to have  a strictly lower expected waiting time but in some cases a lower priority. All mechanisms that assign a strictly higher priority to the jobs with higher bids in the spot market, admit an equilibrium where the spot market bids are increasing in the waiting cost, and have no reserve price satisfy these restrictions.

We now characterize the participation decision facing jobs.

\begin{lemma}
For each class $i$ there is a cutoff $c_i$ below which jobs participate in the spot market and above which they do not.
\end{lemma}

\begin{proof}
 A job participates in the spot market if the payoff is better than its alternative (0 if the spot market is operated in isolation or $\max\{0,v_i - (p + c)/ \mu\}$ if PAYG with price $p$ is available).  The payoff from participation is $v_i - c \wt{w}(c) -\wt{m}(c)$.  Let $c$ be any class that participates.  Taking the case of the spot market in isolation first, if $v_i - c \wt{w}(c) -\wt{m}(c) \geq 0$ then  $v_i - \wh{c} \wt{w}(c) -\wt{m}(c) > 0$ for all $\wh{c} < c$.  Thus, if  a job of class $i$ with cost $c$ participates all lower cost jobs do as well. This argument also implies that if a job with waiting cost $c$ does not participate, then neither does any  job with waiting cost $\wh{c} > c$ . Thus, there is some threshold $c_i$ below which jobs participate and above which they do not.  The argument with PAYG as an option is essentially the same because the minimum possible value of $\wt{w}(c)$ is $1/\mu$, the same as the waiting time under PAYG.
\qed \end{proof}

In order to characterize an equilibrium where jobs use cutoffs $\mb{c}=(c_1,..,c_n)$, we need to analyze the expected waiting time for a job with waiting cost $c$ in the spot market with cutoffs $\mb{c}$. It suffices to characterize some properties of the waiting times for arbitrary choices of cutoffs. Given a queuing system for the spot market, define the waiting time function $w(c; \mb{c})$ as the expected waiting time of a job with cost $c$ when jobs of class $i$ use cutoff $c_i$. Note that we are defining $w$ for arbitrary cutoffs, not just equilibrium ones.  We define $(c'_i,c_{-i})$ to be the vector obtained from $\mb{c}$ by replacing the  $i$th component by $c'_i$.

The following lemma gives the relevant properties of~$w$.
Intuitively, they capture that jobs with higher waiting costs get lower waiting times, if more jobs decide to enter the spot market then waiting times increase, and that entering jobs do not affect the waiting time of jobs with higher waiting costs.
\begin{lemma} \label{lemma:w-properties}
The waiting-time function $w(c;\mb{c})$ is well defined whenever \newline
$(\sum_i \lambda_i F_i(c_i))/(k\mu) < 1$.
  It is an increasing function of $\sum_{i} \lambda_i\pos{F_i(c_i) - F_i(c)}$. In particular, this implies:
\begin{enumerate} [(i)]
\item
$w(c;\mb{c})$ is decreasing in $c$ for $c \in [0, \max_i{c_i}]$, $w(c;\mb{c}) > 1/\mu$ if $c < \max_i c_i$, and $w(c; \mb{c}) = 1/\mu$ if $c \geq \max_i c_i$.
\item
$w(c;\mb{c})$ is increasing in $c_i$ for $c_i \in [0,\mu v_i]$.
\item
For any $t \geq  \wh{c}_j > c_j$  \newline
$w(t;\mb{c}) = w(t;(\wh{c}_j,\mb{c}_{-j}))$.
\end{enumerate}
\end{lemma}

\begin{proof}
The condition $(\sum_{i} \lambda_i F_i(c_i))/(k\mu) < 1$ ensures the queue is stable and hence expected waiting time is finite. 
 Since priority is given to jobs with higher waiting cost, the expected waiting time of a job with waiting cost $c$ increases with the total arrival rate of the jobs whose waiting cost is  higher than $c$, which is equal to $\sum_i \lambda_i\pos{F_i(c_i) - F_i(c)}$. The job with waiting cost greater than or equal to $\max_i c_i$ gets the highest priority and is served immediately with no interruptions. The enumerated properties follow easily.
\qed \end{proof}

Next, we use a characterization similar to the revenue equivalence theorem for auctions \cite{Myerson81} to  show that the expected payment by any job with waiting cost $c$ is uniquely determined by the waiting time function $w$; in particular, it is the same for any spot pricing mechanism.

Suppose that truthful reporting with cutoffs $\mb{c}$ constitutes a BNE for the given spot pricing mechanism. Let $m(c)$ be the expected payment made by a job with waiting cost $c$ (which is independent of its class). For a BNE to exist, the following \textit{incentive compatibility} (henceforth, IC) constraint must hold: for all $\wh{c}, c \leq \max_i c_i$, and any $i$,
\begin{equation} \label{eq:ic}
v_i - cw(c;\mb{c}) - m(c) \geq v_i - cw(\wh{c};\mb{c}) - m(\wh{c}). 
\end{equation}
By analogy with \cite{Myerson81}, the next lemma relates the expected payment with the waiting time function $w$ and shows that the properties of the waiting time function along with the expected payment given by \eqref{eq:payment} ensure that the IC constraint \eqref{eq:ic} is satisfied.

\begin{lemma} \label{lemma:ic-payment}
A necessary condition for \eqref{eq:ic} to hold is:
\begin{equation} \label{eq:payment}
m(c) = \int_0^c w(t;\mb{c})dt - cw(c;\mb{c}).
\end{equation}
Hence, the expected payment by a job with waiting cost $c$ is uniquely determined by the function $w$.
Moreover, Lemma \ref{lemma:w-properties} and \eqref{eq:payment} together satisfy the IC constraint \eqref{eq:ic}.
\end{lemma}
\begin{proof} [Sketch-  following  Myerson\cite{Myerson81}]
Let
  $$\pi(\wh{c}, c) \triangleq v_i - cw(\wh{c};\mb{c}) - m(\wh{c}).$$
 Then for the IC constraint to hold, the maximum  of  $\max_{\wh{c}}\pi(\wh{c}, c)$ must be achieved, and is achieved  at  $\wh{c} = c$.  Since $\pi$ as a function of $c$ is affine, it follows that $\pi(c,c)=\max_{\wh{c}}\pi(\wh{c}, c)$  is convex, and hence is differentiable almost everywhere,
with (right) derivative:
 \begin{equation}
\frac{\partial }{\partial c}\pi(c,c) =-w(c;\mb{c}).
\label{eq:diffpi}
\end{equation}
Integrating between 0 and $c$,  substituting for  $\pi$ and rearranging give the result under our assumption that $m(0) = 0$ (the job with zero waiting cost won't pay anything in the spot market because waiting is costless for it).\qed \end{proof}


Since $w(c;\mb{c})$ is decreasing in $c$ for $c \in [0,\max_i c_i] $, the proof of Lemma \ref{lemma:w-m-monotonicity} can be used to establish a stronger monotonicity result for the  expected payment $m$. 
\begin{lemma} \label{lemma:payment-inc}
Given cutoffs $\mb{c}$, the expected payment $m(c)$ is increasing in $c$ for $c \in [0,\max_i c_i]$.
\end{lemma}
\begin{proof} 
 Suppose that the waiting cost of a job in class $i$ is $c$ and it instead misreports some $\wh{c} \neq c$. The expected payoff under truthful reporting is $\pi(c, c) \triangleq v_i - cw(c;\mb{c}) - m(c)$ and the expected payoff in case of misreport is $\pi(\wh{c}, c) \triangleq v_i - cw(\wh{c};\mb{c}) - m(\wh{c})$. Considering the cases $\wh{c} < c$ and $\wh{c} > c$ separately, using the property that $w(t;\mb{c})$ is decreasing in $t$, and using \eqref{eq:payment}, we can show that $\pi(c, c) - \pi(\wh{c}, c) > 0$.
\qed \end{proof}

\subsection{Revenue and equilibria for isolated markets} \label{sec:isolation}
Next, we analyze PAYG and the spot market each in isolation. 
\subsubsection{PAYG}
First consider PAYG in isolation. If the PAYG price is $p$, a job from class $i$ with waiting cost $c$ obtains an expected payoff $v_i - (p+c)/\mu$ by using a PAYG instance. A job will participate in PAYG if this payoff is nonnegative. Thus, a job from class $i$ participates in PAYG if its waiting cost $c \leq \mu v_i - p$. The effective arrival rate of class $i$ jobs is then $\lambda_i F_i(\mu v_i - p)$ where $F_i(\mu v_i - p) = 0$ if $p \geq \mu v_i$. Each such job uses a PAYG instance for an expected duration of $1/\mu$ and pays $p$ per unit time. Hence, the expected revenue to the cloud service provider per unit time, denoted by $R^{PAYG}(p)$, is:
\begin{equation} \label{eq:payg-rev}
R^{PAYG}(p) \triangleq \frac{p}{\mu}\bigg(\sum_i \lambda_i F_i(\mu v_i - p) \bigg),
\end{equation}
and the optimum revenue is $\max_{p}R^{PAYG}(p)$.
\subsubsection{Spot Market in Isolation}
For the spot market in isolation, 
 denote the cutoffs in this case by $\mb{c}^S$. From \eqref{eq:payment}, the expected payoff of a job from class $i$ with waiting cost $c$ is $v_i - \int_0^c w(t;\mb{c}^S)dt$. A job will participate in the spot market as long as its expected payoff is nonnegative. Hence, the cutoff vector $\mb{c}^S$ must satisfy: 
\begin{equation}
  \label{eq:spot-cutoffs-constraint}
v_i - \int_0^{c} w(t;\mb{c}^S)dt \left\{ 
\begin{array}{l l}
  \geq 0 & \quad \text{if $c < c_i^S$,}\\
  = 0 & \quad \text{if $c = c_i^S$.}\\
\end{array} \right.
\end{equation}
Theorem \ref{thm:spot-equilibrium} below shows that there is an unique cutoff vector $\mb{c}^S$ satisfying \eqref{eq:spot-cutoffs-constraint} which characterizes the BNE for the spot market in isolation. 
\begin{theorem} \label{thm:spot-equilibrium}
The following holds:
\begin{enumerate}[(i)]
\item
There is a unique solution $\mb{c}^S$ to the following system of equations in $\mb{x} = (x_1,\ldots,x_n)$:
\begin{equation}
\label{eq:spot-cutoffs-eqn}
\int_0^{x_i} w(t; \mb{x})dt = v_i.
\end{equation}

\item
In all BNE, a job from class $i$ with waiting cost $c$ participates in the spot market if and only if $c \leq c_i^S$.  
\end{enumerate}

\end{theorem}

%
%
%
%
%
The proof follows by straightforward inductive argument on the number of agent classes, and is given in Appendix \ref{sec:proof-spot-equilibrium}.
To highlight the explicit dependence of the expected payment on the cutoffs vector $\mb{c}^S$, we use $m(c;\mb{c}^S)$; i.e,
\begin{equation} \label{m-cutoff}
m(c;\mb{c}^S) = \int_0^{c}w(t; \mb{c}^S)dt - cw(c; \mb{c}^S).
\end{equation}
Using Theorem \ref{thm:spot-equilibrium}, the expected revenue to the cloud service provider per unit time when the spot market is operated in isolation, denoted by $R^{s}$, is:
\begin{equation} \label{eq:spot-rev}
R^{s} \triangleq \sum_i \lambda_i \int_0^{c_i^S} m(t;\mb{c}^S) f_i(t) dt.
\end{equation}

\subsection{Revenue and equilibria in the hybrid market} \label{sec:together}
We now leverage the insights gained from analyzing PAYG and the spot market each in isolation and move to analyzing the hybrid system where both are operated simultaneously. As mentioned in Section \ref{sec:strategy}, for a given PAYG price $p$, we look for a cutoff vector $\mb{c}(p)$ such that a job from class $i$ with waiting cost $c$ joins the spot market if and only if $c < c_i(p)$, and if so, it reports its waiting cost truthfully; otherwise it joins PAYG as long as $c \leq \mu v_i - p$ (the cutoff for class $i$ if PAYG is operating in isolation). 

A job from class $i$ with waiting cost $c$ gets the expected payoff $v_i - \int_0^{c}w(t; \mb{c}(p))dt$ from using a spot instance and reporting its waiting cost truthfully, while its expected payoff from using a PAYG instance is $v_i - (p+c)/\mu$. It will pick the one which offers a higher expected payoff. If the PAYG price is too high for a class, then no jobs from that class go to PAYG. Theorem \ref{thm:hybrid-equilibrium} below finds the unique cutoff vector $\mb{c}(p)$ and uses it to characterizes the BNE of the hybrid system.
The proof  proceeds along the same lines as that of Theorem~\ref{thm:spot-equilibrium} and is given in   Appendix \ref{sec:proof-hybrid-equilibrium}.  We also derive additional cutoff parameters $\mb{\bar{c}}$ which we require for analyzing the hybrid market, and which are useful as starting points for determining the optimal $\mb{c}^S$.

\begin{theorem} \label{thm:hybrid-equilibrium}
The following holds: 
\begin{enumerate}[(i)]
\item
For each $i \in \{1,2,...,n\}$ there is a unique vector of the form
$$\mb{x} = (x_i ,x_i, \ldots,x_i, x_{i+1}, x_{i+2},\ldots,x_n)$$
that satisfies  $\int_0^{x_j} w(t; \mb{x})dt = v_j$ for all $j \geq i$.  Let $\overline{c}_i$ be the value of $x_i$ in this vector.
\item
Define $v_{n+1} = \overline{c}_{n+1} = \overline{c}_0= 0$ and $v_{0} = \infty$.  Then there is a unique $i\in \{0,1,...,n\}$, denoted $i^*$ that is the class such that $p \in [\mu v_{i^*+1} - \overline{c}_{i^*+1}, \mu v_i - \overline{c}_{i^*})$.
\item
There is a unique solution $\mb{c}(p)$ to the system of equations that 
\begin{equation} \label{hybrid-cutoff-eqn}
\int_0^{x_j} w(t; \mb{x})dt =
\begin{cases}
\frac{p + x_j}{\mu} & \text{if  $j \leq i^*$}\\
v_j & \text{otherwise.}
\end{cases}
\end{equation}
\item
 In any BNE, a job from class $i$ with waiting cost $c$ participates in the spot market if and only if $c < c_i(p)$, it participates in PAYG if $c_i(p) \leq c \leq \mu v_i - p$. If $\mu v_i - p < c_i(p)$ then no class $i$ job participates in PAYG\footnote{It is assumed that jobs break ties between the spot market and PAYG in favor of PAYG.}.  In particular, this is true exactly for the classes such that $i > i^*$.
\end{enumerate}
\end{theorem}

%

%

This theorem also provides insight into the structure of the outcome.  For example, if follows from parts (ii), (iii) and (iv) that 
\begin{corollary}
All classes that participate in PAYG have the same cutoff, $c_{i^*}(p)$, and if the price is set higher than $\mu v_1 - \overline{c}_1 = \mu v_1 - c^S_1$ then no class participates in PAYG and the outcome is the same as if only the spot market existed.
\end{corollary}

Our analysis characterizes a truthful BNE for the system where PAYG and the spot market are operating simultaneously. This equilibrium can be implemented by assigning higher priority to the jobs with the higher waiting cost and collecting the payment according to \eqref{eq:payment}.

The expected revenue to the cloud service provider per unit time is the sum of expected revenue from the spot market and PAYG. From \eqref{eq:payg-rev}, \eqref{eq:spot-rev}, and Theorem \ref{thm:hybrid-equilibrium}, given a PAYG price $p$, the expected revenue per unit time for the hybrid system, denoted by $R^{H}(p)$, is:
\begin{multline} \label{eq:hybrid-rev}
R^{H}(p) \triangleq \sum_i \lambda_i \bigg( \frac{p}{\mu} \pos{F_i(\mu v_i - p) - F_i(c_i(p))} 
+ \int_0^{c_i(p)} m(t; \mb{c}(p))f_i(t)dt \bigg),
\end{multline}
and the optimum revenue is $\max_{p}R^{H}(p)$.
\section{Revenue Comparisons} \label{sec:simulations}
In the previous section, we characterized the equilibrium outcomes and resulting revenue for three different market types.  Now we compare their performance.  Since just having a spot market is a special case of the hybrid market, we focus on whether a cloud service provider should prefer PAYG or a hybrid market.
Perhaps the simplest question we can ask is whether PAYG or hybrid raises more revenue.  The next theorem shows that if the optimal price for the hybrid system is sufficiently small, PAYG in isolation can provide a higher expected revenue to the cloud service provider than operating PAYG and the spot market simultaneously.

\begin{theorem} \label{thm:hybrid-rev-comp}
Suppose the optimal price $p^H$ of the hybrid system is such that $p^H \leq \mu v_n - \overline{c}_n$, i.e., all classes participate in PAYG. Then the optimum expected revenue per unit time from PAYG in isolation is higher than the optimum expected revenue per unit time from the hybrid system; i.e., $\max_{p} R^{H}(p) = R^{H}(p^H) < \max_{p} R^{PAYG}(p)$.
\end{theorem}
\begin{proof}
Since $\max_{p} R^{PAYG}(p) \geq R^{PAYG}(p^H)$,
it suffices to show that $R^{PAYG}(p^H) > R^{H}(p^H).$

Since $p^H \leq \mu v_n - \overline{c}_n$, then for all $i$ and $j$, $c_i (p^H)  = c_j (p^H) \leq \overline{c}_n$, implying $\mu v_i - p^H \geq \overline{c}_n \geq c_i(p)$. Then from \eqref{eq:payg-rev} and \eqref{eq:hybrid-rev},
\begin{multline} \label{eq:hrc-eq2}
R^{PAYG}(p^H) -  R^{H}(p^H) = \sum_i \lambda_i \bigg( \frac{p^H}{\mu} F_i(c_i(p^H)) 
- \int_0^{c_i(p^H)} m(t; \mb{c}(p^H))f_i(t)dt \bigg).
\end{multline}
At $c = c_i(p^H)$, a job is indifferent between PAYG and the spot market. Hence, 
\begin{equation}
c_i(p^H) w(c_i(p^H); \mb{c}(p^H)) + m(c_i(p^H); \mb{c}(p^H)) = \frac{c_i(p^H) + p^H}{\mu}.
\end{equation}
Since $c_i (p^H) = c_j(p^H)$, $w(c_i(p^H); \mb{c}(p^H)) = 1/\mu$. Hence, \newline $m(c_i(p^H); \mb{c}(p^H)) = p^H/\mu$. From Lemma \ref{lemma:payment-inc}, $m(t; \mb{c}(p^H))$ is increasing in $t$ for $t \in [0, c_i(p)]$. This and \eqref{eq:hrc-eq2} imply:
\begin{multline} \label{eq:hrc-eq3}
R^{PAYG}(p^H) -  R^{H}(p^H) > \\
\sum_i \lambda_i \bigg( \frac{p^H}{\mu} F_i(c_i(p^H)) - \int_0^{c_i(p)} \frac{p^H}{\mu} f_i(t)dt \bigg) = 0.
\end{multline}
\qed
\end{proof}

The intuition behind Theorem \ref{thm:hybrid-rev-comp} is that, in a class that participates in PAYG, all jobs of that class that instead choose the spot market would prefer PAYG to balking.  Since they pay less money (but more waiting time) to use the spot market, we could make more money if we could prevent them from entering the spot market.  When this is true of every class, we can actually prevent them, by simply eliminating the spot market.

Obviously Theorem \ref{thm:hybrid-rev-comp} has significantly more bite with a small number of classes, since it requires that all participate in PAYG.  However, we note that for Amazon only a small percentage of jobs are submitted to the spot market, so this may well be the relevant case.  Further, we conjecture that the revenue ranking result holds much more broadly.  As an example, we simulate the performance of a spot market that consists of $k$ parallel $M/M/1$ queues with two classes, where  jobs bid for preemptive priorities.
An arriving job is randomly and uniformly sent to one of the $k$ queues where it is served according to its priority order, determined by its bid, in that queue. 
As shown in prior work~\cite{Lui1985},
the waiting time is given by 
\begin{equation}\label{eq:w-example}
w(c;c_1,c_2) = \frac{1}{\mu\big(1- \displaystyle \sum_{i = 1,2}\rho_i\pos{F_i(c_i)- F_i(c)}\big)^2},
\end{equation}
where $\rho_i \triangleq \lambda_i/(k \mu)$.
The proof of Theorem \ref{thm:hybrid-equilibrium} provides a recipe for numerically computing the cutoff vector $\mb{c}(p)$ as a function of PAYG price $p$.
  
We randomly generated one hundred random configurations of the values of $v_i$'s, $\lambda_i$'s, and $k$.
All were chosen uniformly at random from $[0,20]$, with $k$ a random integer from this range.
The service rate $\mu$ was kept constant at one and $F_i$ was uniform in the interval $[0, \mu v_i]$.


 For each realized configuration, we observed that the optimal revenue from PAYG in isolation was always higher than the optimal revenue from the hybrid system where PAYG and the spot market are operating simultaneously, even for the case where the optimal price $p^H$ of the hybrid system is greater than $\mu v_2 - \overline{c}_2$.
An example plot where $p^H > \mu v_2 - \overline{c}_2$
is shown in Figure \ref{fig:payg-spot-plot}%
\footnote{The exact parameter values are
$v_1 = 10.508088077186715$,
$v_2 = 1.876400535497311$,
$\lambda_1 = 0.7576345977040905$,
$\lambda_2 = 1.2301997305619036$,
and
$k = 1$.}.
Observe that if PAYG price is low, most of the jobs in the hybrid system use PAYG and pay a small price, leading to a small expected revenue. As PAYG price increases, jobs move to the spot market,
reaching a point where all jobs use the spot market.
At $p=\mu v_2$, all class $2$ jobs balk from PAYG leading to a kink in the plot for PAYG in isolation. However, it is not the case that PAYG raises more revenue at all prices $p$, as near this kink the hybrid market raises slightly more revenue.
Additional simulations with exponential and beta distributed waiting costs also failed to generate instances where the hybrid market generated more revenue.  This suggest that for a broad range of parameters not operating a spot market is optimal.
\begin{figure}[t] 
\begin{center}
\includegraphics[trim=0.5in 2.75in 0.50in 3.0in, clip=true, height=2.6in]{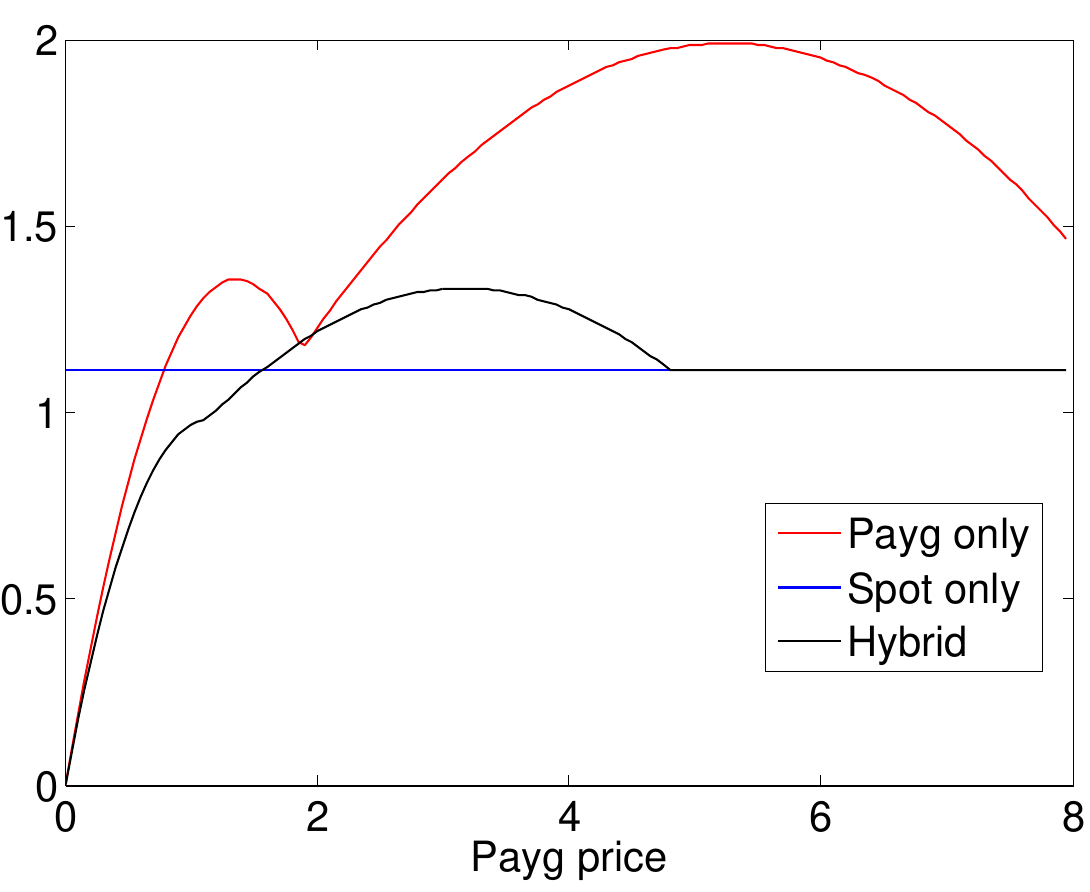}
\caption{Expected revenue $R^{PAYG}(p), R^H(p), R^S(p)$ for PAYG, Hybrid and Spot respectively  as a function of PAYG price $p$ for fixed $\lambda_1\approx0.76, \lambda_2\approx1.23, v_1 \approx 10.5, v_2 \approx 1.88$ with $k=1$ and uniform waitings cost distributions.\label{fig:payg-spot-plot}} 
\end{center}
\end{figure}

We now explore a situation where  intuitively a hybrid scheme should outperform PAYG,  specifically where there are two types of jobs:  high value jobs  with high waiting time costs,  and   low value jobs having low waiting time costs.  A hybrid mechanism where high value jobs use PAYG, and low value use the spot market might be expected to generate more revenue than setting a single PAYG price.   We shall see that this is still often not the case, essentially for two reasons:
\begin{enumerate}
\item  In any hybrid scheme, if the price is such that only the high class participates in PAYG ($i^*=1$) some high value traffic will always use the Spot Market.  This is because it follows from \eqref{hybrid-cutoff-eqn}  in Theorem~\ref{thm:hybrid-equilibrium} that $c_1(p) > 0$, hence $F_1(c_1(p))>0$.
\item   Under the same setting, not all the low value traffic can use the spot market:  by the proof of Lemma~\ref{lem:unique-solution} in the appendix $c_2(p)$ is decreasing in $p$ in this range.  Thus  $c_2(p)<v_2$, and hence $F_2(c_2(p))<1$.
\end{enumerate}

In other words we cannot achieve perfect separation of the two classes by price alone.  We use 
Beta distributions to examine this example in more detail. Consider a setting where there is one type with high value ($v_1 = 10$) and high waiting costs while another with low value $v_2$ and low waiting costs.  Figure~\ref{fig:pdfs} shows such an example where the waiting cost for the low value type is essentially all between 0 and 1 and the waiting costs for the high value type are essentially all slightly less than 5.  The other parameters are $k = 10$ and $\lambda_1 = \lambda_2 = 5$.   We now explore what happens as we vary $v_2$ between 3 and 4, settings where the optimal PAYG price excludes type 2 traffic, and hence Theorem \ref{thm:hybrid-rev-comp} does not apply.
Intuitively, we could set a PAYG price close to 5 to try and capture as much revenue from the high values as possible while getting some revenue from the low values.
Figure~\ref{fig:beta-rev} shows that when $v_2=3$  even in this setting PAYG alone is still optimal (though only very slightly).  The essential problem is that the high PAYG price gives class 1 jobs an incentive to drop down into the spot market, wiping out some of the gains from serving the class 2 jobs.  It does suggest that the hybrid system can be more robust in some cases if there is uncertainty about how best to set prices.  However, this is an example created to make the hybrid system look as good as possible; in Figure~\ref{fig:payg-spot-plot}, which is much more representative of the examples generated in our simulations, the hybrid system does not add such robustness.

If we increase $v_2$, putting $v_2=3.5$, then as shown in Figure \ref{fig:beta-rev2},  the hybrid market does indeed generate (slightly) more revenue that PAYG.    But if $v_2$ is increased still further to $v_2=4$, then once again PAYG generates more revenue.
For this parameterized model, the hybrid model is only better if approximately $3.1 \leq v_2 \leq 3.8$, whereas PAYG generates more revenue if approximately $v_2 \leq 3$ or $v_2 \geq 3.9$. Similarly, if  we fix all parameters except $\lambda_2$, there is a small interval $\mathcal{I}$,  when  hybrid is optimal for $\lambda_2 \in \mathcal{I}$, and PAYG optimal for all other $\lambda_2$.  Thus, even in an example designed to make the hybrid mechanism look as good as possible, there is a relatively small range of parameters where it is superior.
\begin{figure}[t] 
\begin{center}
\includegraphics[trim=0.75in 3.25in 0.75in 3.25in, clip=true, height=2.6in]{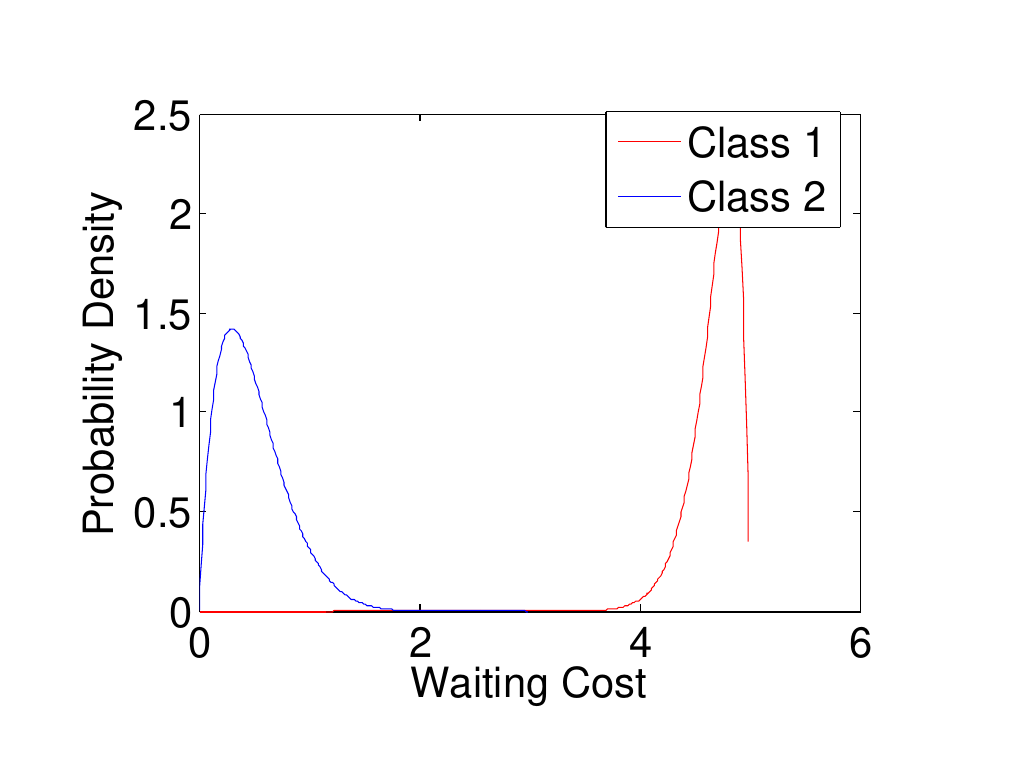}
\caption{Exanple pdfs $f_i(c)$ of well-separated waiting costs.\label{fig:pdfs}} 
\end{center}
\end{figure}

\begin{figure}[t] 
\begin{center}
\includegraphics[trim=0.75in 3.25in 0.75in 3.25in, clip=true, height=2.6in]{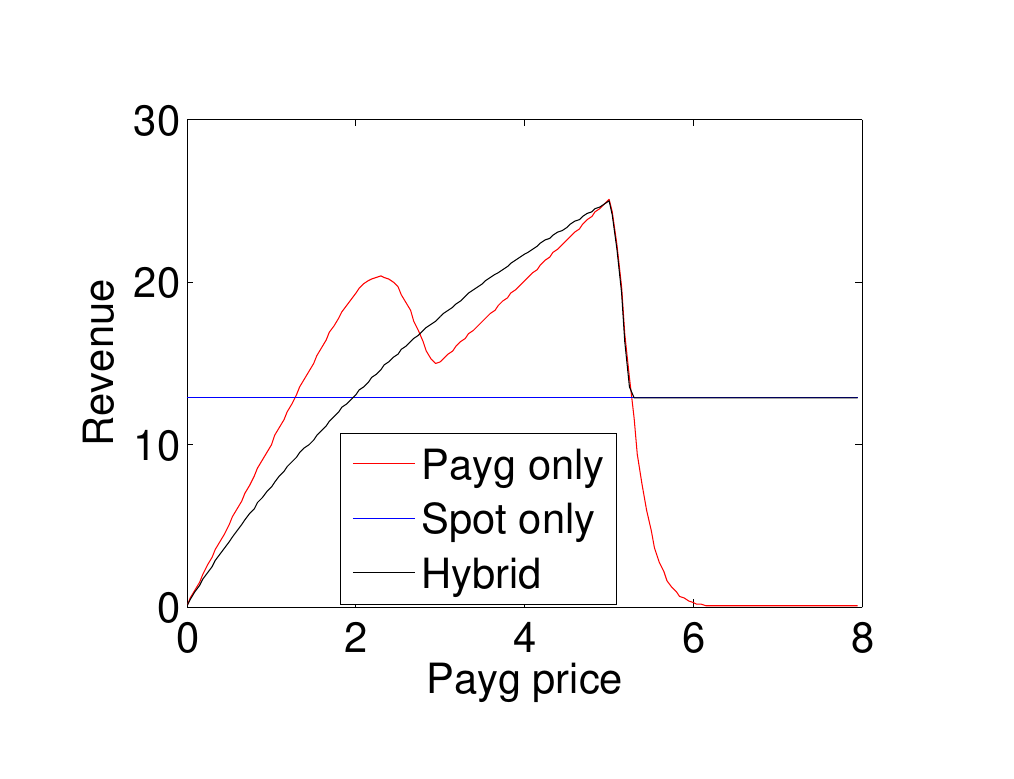}
\caption{Expected revenue $R^{PAYG}(p), R^H(p), R^S(p)$ for PAYG, Hybrid and Spot respectively  as a function of PAYG price $p$ for fixed $\lambda_1= \lambda_2=5$, $k=10$,  $v_1=10$ and $v_2=3$. Here  PAYG is  (just) optimal.  Waiting cost pdfs $f_i(c)$ as in Figure \ref{fig:pdfs}.  \label{fig:beta-rev}} 
\end{center}
\end{figure}

\begin{figure}[t] 
\begin{center}
\includegraphics[trim=0.75in 3.25in 0.75in 3.25in, clip=true, height=2.6in]{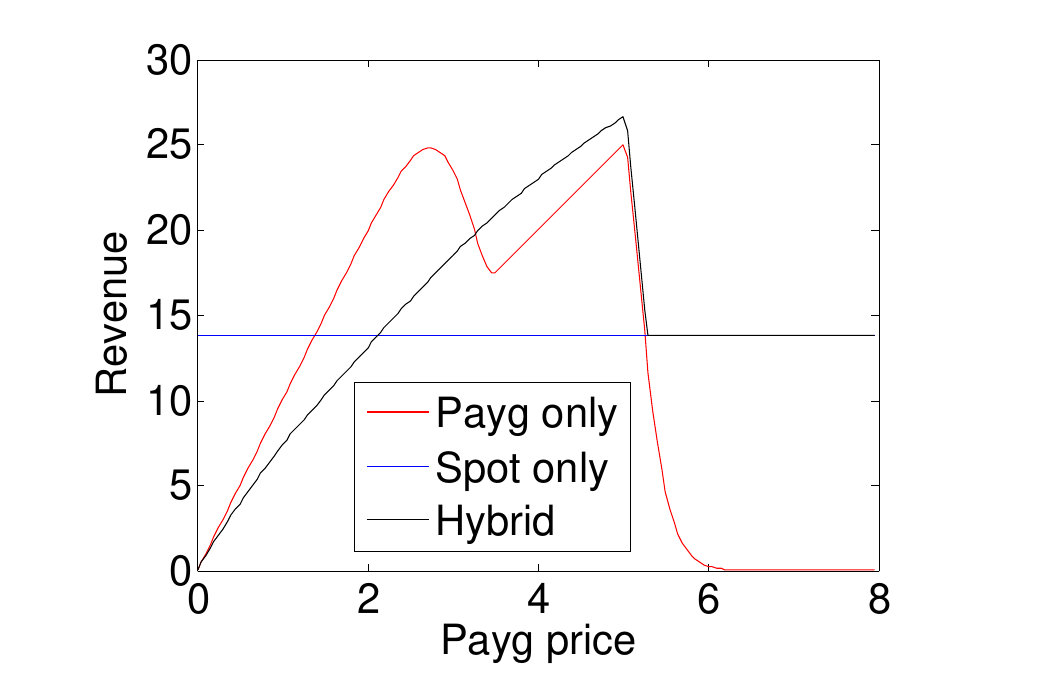}
\caption{Expected revenue $R^{PAYG}(p), R^H(p), R^S(p)$ for PAYG, Hybrid and Spot respectively  as a function of PAYG price $p$ for fixed $\lambda_1= \lambda_2=5$ , $k=10$,  $v_1=10$, as in Figure \ref{fig:beta-rev}, but with  $v_2=3.5 $. Now the Hybrid scheme is optimal. \label{fig:beta-rev2}} 
\end{center}
\end{figure}

\section{Welfare Analysis}
\label{sec:welfare}

The two systems also provide different social welfare (i.e. the value of served jobs minus their waiting costs). For classes that participate in both PAYG and spot, agents that send their job to the spot market incur a higher waiting cost, which reduces welfare.  On the other hand, classes that do not participate in PAYG receive some service in a spot market, which increases welfare.  If the service provider wishes to raise a given amount of revenue, Figure~\ref{fig:payg-spot-plot} shows that he would pick a different (lower) price, which increases welfare.

To help characterize social welfare, we now show that the outcome of the hybrid system is, in a sense, efficient.
Economic efficiency is the property of maximizing social welfare, the total utility in the system.  Since payments are just a transfer of utility from one agent (the job) to another (the owner of the system), they are irrelevant.
It has previously been observed for restricted cases that spot market outcomes are efficient (e.g.~\cite{Hassin1995}).  We show that this is true in general, but only if we treat the PAYG price as a real cost.  Hence we say that the outcome is \emph{pseudo-efficient}, because it would be efficient if the PAYG price represented a real cost.

\begin{theorem}
\label{thm:tradeoff}
The equilibrium of the hybrid system is pseudo-efficient: given the PAYG price $p$, agents make the socially optimal decision in determining whether to send their job to the spot market, send it to PAYG, or balk.  Equivalently, this mechanism implements the VCG outcome.  Thus, the payment of an agent is his (expected) externality 
\end{theorem}

\begin{proof}
Suppose we use the VCG mechanism, which selects the pseudo-efficient outcome and charges each agent his externality.  In this outcome a given class of job must have all jobs with waiting costs below a cutoff sent to the spot market.  Otherwise, jobs with a lower waiting cost that are not sent to it could be swapped with jobs with a higher waiting cost, increasing social welfare.  Consider the marginal effect of admitting a job at the cutoff to the spot market for any class.  Since the outcome is pseudo-efficient, the externality this causes plus its own waiting cost must be equal either to the cost it would experience under PAYG (if it would otherwise go there) or to its value (if it would otherwise balk).  Since VCG is incentive compatible, we know that its payment (equal to the externality) satisfies \eqref{eq:payment}.  But then the equations that define the cutoffs of the spot market are exactly the same as those from \eqref{hybrid-cutoff-eqn}, whose unique solution is the outcome of the hybrid system.
\qed
\end{proof}

\begin{figure}[t] 
\begin{center}
\includegraphics[trim=-0.25in 2.75in 0.50in 3.0in, height=2.6in]{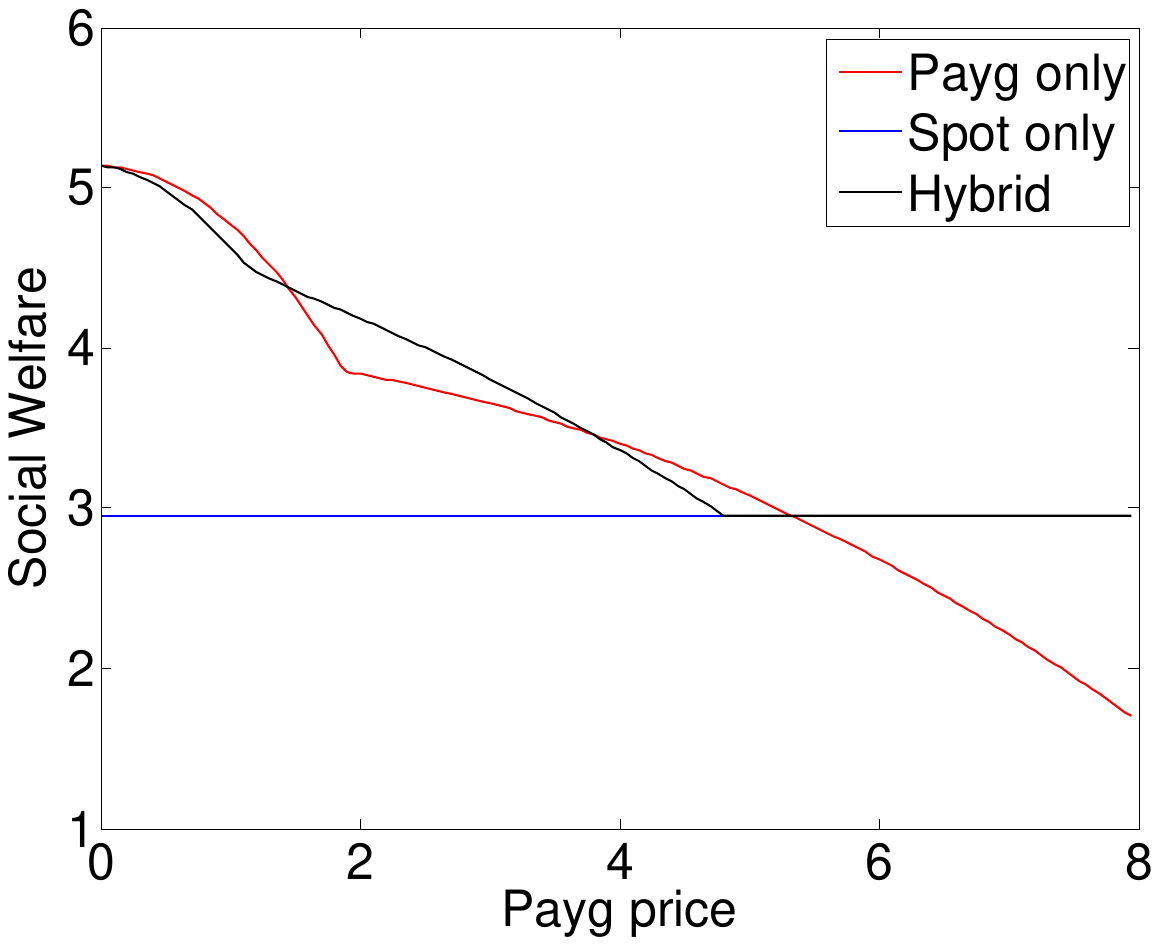}
\caption{Social welfare $W^{PAYG}(p), W^S(p), W^H(p)$  as a function of PAYG price for Figure \ref{fig:payg-spot-plot} parameter settings .\label{fig:payg-spot-sw}} 
\end{center}
\end{figure}

\begin{figure}[t] 
\begin{center}
\includegraphics[trim=-0.25in 2.75in 0.50in 3.0in,  height=2.6in]{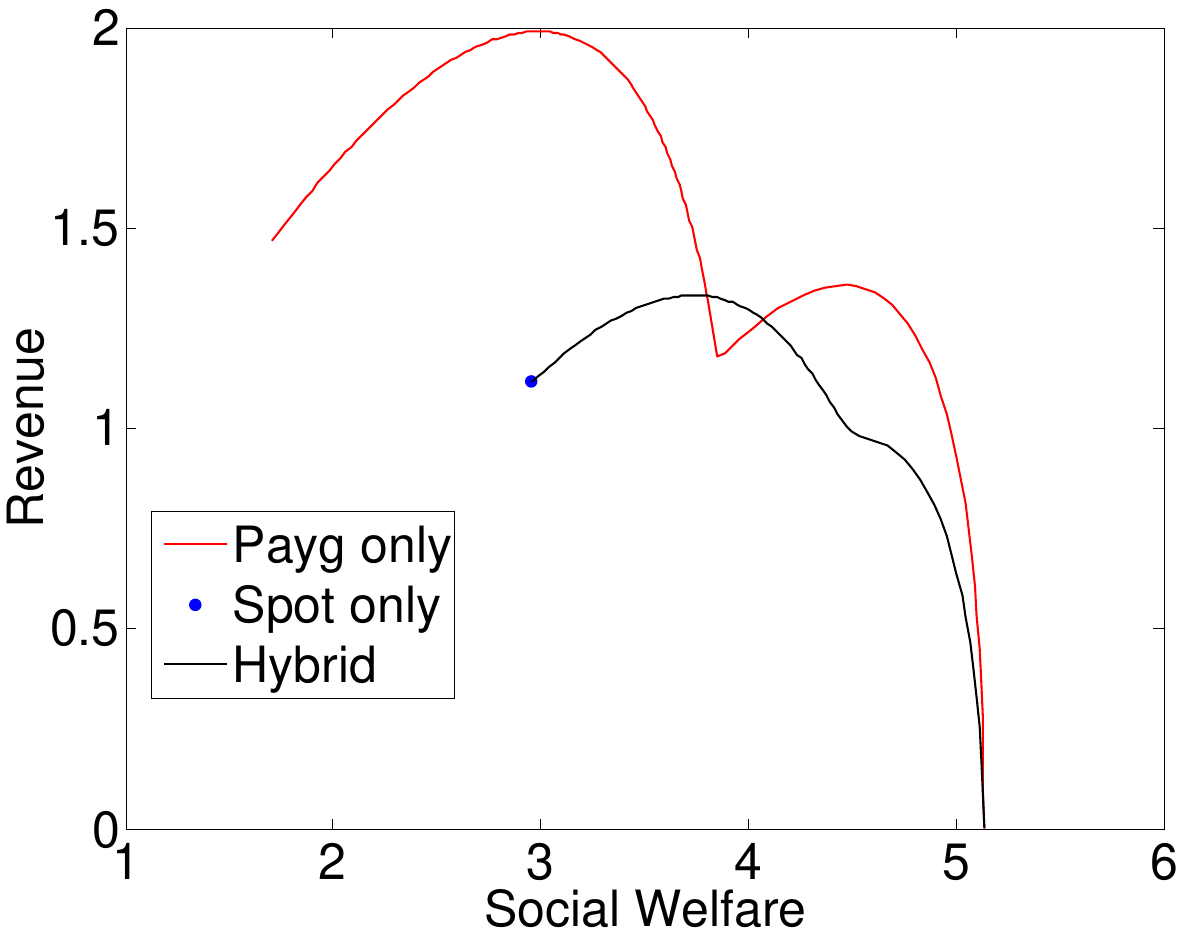}
\caption{Tradeoff between revenue $R(p)$ and welfare  $W(p)$   for Figure \ref{fig:payg-spot-plot} parameter settings.\label{fig:payg-spot-tradeoff}} 
\end{center}
\end{figure}

Given a PAYG price $p$, the expected social welfare per unit time for the hybrid system, denoted by $W^{H}(p)$, is:
\begin{multline} \label{eq:hybrid-welf}
W^{H}(p) \triangleq \sum_i \lambda_i \bigg( \int_{c_i(p)}^{\mu v_i - p}\left(v_i - \frac{t}{\mu}\right )f_i(t) dt   
+ \int_0^{c_i(p)}\left (v_i -  w(t; \mb{c}(p)\right )f_i(t)dt  \bigg).
\end{multline}

Figure~\ref{fig:payg-spot-sw} illustrates the effect of prices on social welfare in the same example from Figure~\ref{fig:payg-spot-plot}.  For much of the range of prices, the welfares of PAYG and the hybrid sytem are similar, with the better solution changing several times.  Once prices are high enough that no one participates in PAYG in the hybrid system, social welfare in the hybrid system is higher; because PAYG tends to extract more revenue at a given price, Figure~\ref{fig:payg-spot-tradeoff} shows that, at least in this example, PAYG enables a better tradeoff interesting most operating points.  In particular, while there are choices of social welfare for which the hybrid market raises more revenue, for every such choice there is a choice with higher social welfare where the hybrid market simultaneously achieves more revenue (i.e. a Pareto improvement in economic terminology).  So in this example, PAYG not only raises more revenue but enables better tradeoffs between revenue and social welfare.  Gaining a theoretical understanding of these tradeoffs is an important question for future work.

\section{Relaxing Model Assumptions}
%
We now consider extensions to our model, which relax some of our assumptions.

\subsection{Finite Capacity PAYG}
\label{sec:finiteCap}
We have modeled PAYG as infinite capacity system, whereas the spot market is has a finite capacity ($k$) servers.  This dichotomy is intended to represent current cloud-computing markets, where IaaS is primarily sold
using a PAYG model with availability guarantees, whereas the spot market is much smaller.  However, we believe that with respect to our results, the assumption that PAYG is modeled as a $GI/GI/\infty$ queue, rather than (say) a $GI/GI/n/n$ queue
is essentially without loss of generality.

If PAYG has a finite capacity, then in order to give service level guarantees, there is some arrival rate $\lambda^*$ and associated load $\rho^*$ such that the  performance of PAYG meets service level guarantees for all $\lambda \leq \lambda^*$.  For example, by choosing $\lambda^*$ such that $\Pr\{\mbox{arriving jobs sees all $n$ servers busy} |\lambda=\lambda^*\}\leq \epsilon$, for some small $\epsilon\approx 0$.    Provided $\lambda \leq \lambda^*$,  the behavior of the PAYG queue will approximated by $GI/GI/\infty$ queue.
 For $\lambda>\lambda^*$,  with a finite capacity queue, certain jobs would not be admitted (c.f \cite{AgmonBen-Yehuda:2013} for some  information on Amazon EC2 PAYG policies).
Thus, we could incorporate this in our model by assuming that there is a small probability that jobs submitted to PAYG are not served and pay cost 0.  This slightly reduces revenue from PAYG and the PAYG part of hybrid and results in classes choosing slightly higher cutoffs in Theorem~\ref{thm:hybrid-equilibrium}, but does not materially affect our results.  

In more detail,  let $B(\rho,n)$ denote the probability that an arriving job sees all servers busy in an  $GI/GI/n/n$ queue, which serves PAYG customers.  Define $\rho^{PAYG}(p)=\frac{1}{\mu}\sum_i \lambda_i F_i(\mu v_i - p) $, and let $\epsilon=B(\rho^{PAYG}(p),n)$.   Then Theorem~\ref{thm:hybrid-equilibrium} holds as stated, but with the BNE in part (iv) replaced by $\epsilon^*$-BNE, where $\epsilon^*$ is upper bounded by $ v_i (1- p/ \mu) \times  \epsilon$.
 Compared to $n=\infty$, $R^{PAYG}(p)$  decreases by $p \rho^{PAYG}(p) \times B(\rho^{PAYG}(p),n)$, while $R^H(p)$ decreases by $p \rho^{H}(p) \times  B(\rho^H(p),n)$,   where $\rho^S(p)\triangleq \frac{1}{\mu} \sum_i \lambda_i F_i(c_i(p))$, $\rho^H(p)=\rho^{PAYG}(p)-\rho^S(p)$;  both decrements are $O(\epsilon)$ under our assumptions that $\rho^s \ll \rho^{PAYG}$.

Alternatively Theorem 2 can be amended to allow BNE rather than $\epsilon^*$-BNE on replacing \eqref{hybrid-cutoff-eqn} with 
\begin{equation*}
\int_0^{x_j} w(t; \mb{x})dt =
\begin{cases}
\frac{p + x_j}{\mu} (1- B(\rho^H(p),N)+ v_j B(\rho^H(p),N) & \text{if  $j \geq i$}\\
v_j  & \text{otherwise,}
\end{cases}
\end{equation*}
and using this to define the unique $\mb{c}(p)$  and  $\overline{c}_{i^*}$  via (ii). The finite capacity 
results in an $O(\epsilon)$  increase in the cut-offs.   Notice that now there is weak $O(\epsilon)$ dependence of the cut-offs on the class, and hence payments also will have this weak class dependence.  However, these do not affect our results that in general  PAYG raises more revenue.

\subsection{Reserve Prices}
\label{sec:reserve}

In our basic model, we assumed that $m(0) = 0$.  This means that there is no minimum payment (reserve price), so jobs willing to tolerate a long enough wait can be served for free.  Imposing such a reserve price is typically a way to raise revenue, so intuitively adding them to the spot market seems beneficial and potentially a way to make the hybrid market more profitable than PAYG alone.  Indeed, it is at least weakly better because a PAYG is the special case of a hybrid market with a reserve price equal to the PAYG price.  In this section, we explain why this extra option need not actually help.

First, we explain how our characterizations from Theorems~\ref{thm:spot-equilibrium} and~\ref{thm:hybrid-equilibrium} change in the presence of reserve prices.  If we institute a reserve price of $m(0) = r$, then the incentive compatibility constraint given by Lemma~\ref{lemma:ic-payment} becomes
$$m(c) = r + \int_0^c w(t;\mb{c})dt - cw(c;\mb{c}).$$
With this change, we need to change the systems of equations based on indifference between the spot market and PAYG or balking.  Specifically, \eqref{eq:spot-cutoffs-eqn} becomes
$$\int_0^{x_i} w(t; \mb{x})dt = v_i - r,$$
while \eqref{hybrid-cutoff-eqn} becomes
\begin{equation*}
\int_0^{x_j} w(t; \mb{x})dt =
\begin{cases}
\frac{p + x_j}{\mu} - r & \text{if  $j \geq i$}\\
v_j - r & \text{otherwise.}
\end{cases}
\end{equation*}
The same proofs go through, {\em mutatis mutandis}.

This new characterization, shows that adding a reserve price doesn't change anything significant about the behavior or jobs.  The spot market is now expensive, so fewer jobs use it with the rest either choosing PAYG or balking, but this structure is still determined by cutoffs as before.  Thus it is still the case that, for every class with jobs that participate in both PAYG and the spot market, all those jobs would have been willing to pay the higher PAYG price.  This is the main insight behind Theorem~\ref{thm:hybrid-rev-comp}, so the proof of that theorem goes through as well.  Thus, if the optimal combination of PAYG price and reserve price for a hybrid system is such that all classes of jobs participate in PAYG, then the reserve price is equal to the PAYG price and the system is effectively PAYG alone.

Furthermore, even in ranges of parameters where a hybrid market is superior, this analysis suggests there is an option that is better than either: offer only PAYG, but with a menu of prices where lower prices receive a larger (artificial) delay before (an approach explored by Af\`{e}che~\cite{Afeche2004}).  By way of intuition if the menu of prices and delays are chosen to be the payments and expected waiting time at the various cutoff values for different classes in the spot market, all jobs of a class that were making less than their cutoff payment would still be willing to pay the higher amount rather than balk.  Indeed, in the example of Figure~\ref{fig:beta-rev2} where a hybrid market can raise revenue from 25.0 to 26.6, such a scheme can generate revenues upwards of 33.  Such artificial delays are not without precedent: among other restrictions the Amazon Glacier cloud storage system may delay responding to requests for up to 3 to 5 hours.
Further, subsequent to our work Google launched ``Preemptible VMs'' which take exactly this approach of offering a fixed discount for the risk of preemption. 

This analysis is in stark contrast to that of Doroudi et al.\cite{Doroudi-et-al2013}, who found that a spot market with a reserve price is optimal in their setting (and in particular superior to PAYG).  However, their result relies on PAYG having the same resource constraints as the spot market does, a feature we have argued is not present in current cloud systems.

These results are under our Incentive Compatible allocation assumption, which implies that jobs are able to discover enough information from the system for it to be optimal to bid truthfully.   This implies that any reserve price must be public knowledge or discoverable.   Ben-Yehuda et al.~\cite{AgmonBen-Yehuda:2013} argue that Amazon EC2 spot prices appears to use  a dynamic or random reserve price.  Equivalently, this result could be due to having a limited supply of spot instances that is variable over time.  Such opacity can cause agents (jobs) to be unsure what is their best report, thus making the mechanism (and pricing) not incentive compatible. 

\subsection{Spot Market Costs}

Another variant of our model would be to assume that participation in the spot market is costly.  In practice, designing systems to backup and resume work from checkpoints may require additional effort, so being preempted may not actually be costless the way our basic model assumes.  If all classes of jobs must pay the same cost to participate in the spot market, then this essentially serves as a reserve price, except that it simply represents a loss of efficiency rather than an increase in revenue.  If the cost is not the same for all classes, the situation is more interesting.  In particular, if the cost is higher for classes with higher values, this opens up an opportunity to discriminate between classes that can make a hybrid market clearly profitable.  Such ``damaged goods'' approaches to market segmentation are common across a variety of markets~\cite{McAfee07}.

Suppose there are only two classes with $v_1 > v_2$ and class 1 is required to pay a cost of $s_1$ to participate in the spot market.
Our results are easily adapted to the general  case of indirect spot-market participation costs: for example, in 
 \eqref{eq:spot-cutoffs-constraint} and in Theorem \ref{thm:spot-equilibrium} $v_i$ is simply replaced by $v_i-s_i$.
However, now it is not necessarily the case that $c_1(p) > c_2(p)$.  To see this, take the extreme case of $s_1 = v_1$, in which case jobs of class 1 will not participate in the spot market while jobs of class 2 will still do so.  Thus, a high PAYG price could be set to optimally extract revenue from jobs of class 1, and even if this price is higher than $\mu v_2$ some revenue would still be extracted from the jobs of class 2 in the spot market with no jobs of class 1 defecting to it.

As related differentiation occurs when the support of the waiting cost is $[s_i, \mu v_i]$ rather than $[0,\mu v_i]$.  Our analysis goes through using obvious alterations to the lower limits of integrals (equivalently, putting $f_i(t)=0$ for $t \in [s_i,0$]).  The quantitative behavior is similar to when $s_i$  is an indirect cost.  For example, in the two-class example, if $s_1=v_i$ and $s_2=0$ then class 1 jobs will not participate in the spot market.
\section{Discussion and Future Work} \label{sec:discussion}

Our analysis characterizes a truthful BNE for the system where PAYG and the spot market are operating simultaneously.  Our theoretical results show that in many cases the revenue raised by a PAYG system in isolation with a well chosen price $p$ dominates that of this hybrid system.  In particular, we have proved that this always holds true when all classes participate in the PAYG market. It will also hold true in the degenerate case, where the PAYG is chosen suboptimally high, so that all cases prefer to enter the spot market.    Simulations suggest that this may be true in general,  regardless of whether individual classes of  enter the PAYG market.

Our results contrast with previous work, and may also appear counterintuitive to those expecting price discrimination to automatically yield higher revenues.  One significant difference between our work and previous analyses is the combination of two different markets operating simultaneously and under different mechanisms, one an incentive-compatible mechanism and the other a fixed price design requiring no information from the users.   We specifically chose such a system as an abstraction of current pricing systems.  If there is but a single mechanism for all users, that is a single scheduling and pricing  design, then optimal dynamic or auction-based  pricing can raise move revenue than fixed pricing, a result that has been proved under a variety of assumptions. 
   The introduction of a secondary market changes the picture by changing user incentives.   With the combined system, it is very difficult to extract extra revenue by using a spot market, because of the difficulty of avoiding cannibalizing the primary market, where low waiting cost, high value customers choose to drop to the (cheaper)  spot market, thus decreasing revenue.  Our results frame this more precisely.    
  
 Our analysis is based on a number of assumptions.  However, as we have shown, many of them are not critical for our conclusions to hold.  For example, we model the PAYG system as having infinite capacity, which we believe is reasonable given that capacity is endogenous and PAYG jobs are more profitable than spot market jobs.  However, this can be relaxed as long as the capacity is ``large.''  Similarly, our analysis is robust to the ability to set reserve prices in the spot market.

One assumption that does affect our findings is the assumption that the only indirect costs to jobs are waiting time, and hence that other indirect costs are zero, costs such as those associated with preemption or rewriting applications to enable them to cope with possible preemption.   If these additional costs,  $s_i$, differ among classes and are non-zero, it need no longer be the case that the cut-offs $c_i$  decrease with   (increasing) $i$.  Indeed,  $s_i$ may be sufficiently large to offset the adverse selection problem, causing instead high value jobs to stay with PAYG so that the hybrid market extracts more revenue.  For example, having both  $v_i$ and $s_i$  decreasing with $i$ but differences $v_i-s_i$  increasing implies high value jobs favour PAYG while low value favour a spot market.

Lastly, there are three assumptions which merit further study: first, the current analysis assumes that the arrival process is independent of job type. This may not be true if both arrival pattern and value depend on underlying characteristics of the job.  An example of this would be if more valuable jobs tend to arrive at certain times of day. 
Then it is possible that there are equilibria where jobs of different classes but the same cost have different outcomes. Though  as both classes have the same set of optimal outcomes, this requires an amount of coordination on tiebreaking that may be unreasonable in practice. Secondly, our framework is  for a monopolistic provider. The effect of competitive pressures needs to be investigated.  Third, we have assumed that waiting costs are linear.  This does make the equilibrium analysis dramatically easier because only the expected waiting time matters rather than the full distribution, but does mean that the waiting cost can exceed the value of the job, resulting in negative utility.
  
We conclude by discussing the important point that Amazon does in fact operate a spot market, despite our results suggesting that it may not be optimal from a revenue perspective to do so.  One possibility is that their arrival distribution happens to be in one of the ranges where this is in fact optimal or some assumption our model makes is not applicable in their setting.  However, another possibility is that it is being used for reasons other than revenue optimality.  For example, even if Amazon is making less money, they may be gaining useful information about what jobs can easily be interrupted if the system experiences an unexpected spike in demand or large-scale failure.  Alternatively, the lack of revenue optimality may be exactly the point if the spot market is viewed as a way to gain new customers by offering them lower prices while they are still operating at a smaller scaler.  This is consistent with the observation that, anecdotally, Amazon makes it difficult to operate in the spot market at a large scale.  Finally, based on the analysis of Ben-Yehuda et al.~\cite{AgmonBen-Yehuda:2013} which found Amazon controls reserve prices and causes them to spike, Amazon may actually be using something closer to the menu pricing approach we discuss in Section~\ref{sec:reserve}.  Perhaps tellingly in this regard, when Amazon introduced Glacier as a less expensive storage service, they adopted artificial delays rather than a sport market for data access.  Thus, we do not view Amazon's operation of a spot market as necessarily contradicting our model or results.

\appendix
\section{Proof of Theorem \lowercase{\ref{thm:spot-equilibrium}}} \label{sec:proof-spot-equilibrium}
%
%


We begin with a technical lemma  that will be used several times in the appendix.

\begin{lemma}
\label{lem:unique-solution}
Let $(x_1,\ldots,x_k$ and $g_i(x_i),\ldots,g_n(x_n)$ be given such that $g_j(x_j)$ is weakly increasing and semidifferentiable%
\footnote{A function is semidifferentiable if it has left and right derivatives but they need not be equal.  Note that this implies continuity.  An example relevant for Theorem~\ref{thm:hybrid-equilibrium} is a continuous piecewise linear function.} with left derivative at most $1 / \mu$.
Then there exists unique $x_i,\ldots,x_n$ such that
$$\mb{x} = (x_1,\ldots,x_k,x_i,\ldots,x_i,x_{i+1},\ldots,x_n)$$
satisfies
\begin{equation}
\label{eq:spot-cutoffs-eqn-repeat}
\int_0^{x_i} w(t; \mb{x})dt = g_i(x_i)
\end{equation}
for all $j \geq i$.
\end{lemma}

\begin{proof}
Suppose this holds for $i+1$ to prove that it holds for $i$.\
%
%
By our induction hypothesis, 
given any $z \in [0, \mu v_i]$, there is a unique
$$\mb{x}(z) = (x_1,\ldots,x_k,z,\ldots,z,x_{i+1},\ldots,x_n)$$
satisfying $\int_0^{x_j} w(t; \mb{x}(z))dt = g_j(x_j)$ for all $j \geq i+1$.  Next we show that $\phi(z) \triangleq \int_0^{z} w(t; \mb{x}(z))dt$ is strictly increasing in $z$ for $z \in [0, \mu v_i]$. Since $\int_0^{x_{i+2}} w(t; \mb{x}(z))dt$ is increasing in each $x_j$, $\mb{x}_{i+2}(z)$ is decreasing in $z$. Consider $\wh{z} > z$. Then, $[\mb{x}_{i+2}(z),z] \subset [\mb{x}_{i+2}(\wh{z}), \wh{z}]$, and
\begin{align*}
& \int_0^{\wh{z}} w(t; \mb{x}(\wh{z}))dt - \int_0^{z} w(t; \mb{x}(z))dt \\
& = \int_{\mb{x}_{i+2}(\wh{z})}^{\wh{z}} w(t; \mb{x}(\wh{z}))dt - \int_{\mb{x}_{i+2}(z)}^{z} w(t; \mb{x}(z))dt \\  
& = \int_{t \in [\mb{x}_{i+2}(\wh{z}), \wh{z}] \backslash [\mb{x}_{i+2}(z),z]}  w(t; \mb{x}(\wh{z}))dt \\
& \qquad +\int_{\mb{x}_{i+2}(z)}^{z}\left( w(t; \mb{x}(\wh{z})) -  w(t; \mb{x}(z))\right)dt \\
& >\int_{\mb{x}_{i+2}(z)}^{z}\left( w(t; \mb{x}(\wh{z})) -  w(t; \mb{x}(z))\right)dt > 0,
\end{align*}
where the last inequality follows from Lemma \ref{lemma:w-properties}. Hence, $\phi(z)$ is increasing.

Since $\phi(0) = 0 \leq g_i(0)$, $\phi(\mu v_i) \geq v_i \geq g_i(\mu v_i)$, $\phi(z)$ is differentiable with $\phi'(z) > 1 / \mu$, and $g_i$ is weakly increasing and semidifferentiable with left derivative at most $1 / mu$, there is a unique $z$ solving $\phi(z) = v_{i+1}$. This establishes the claim.
\end{proof} 


Applying the lemma immediately gives the first part of the theorem.  We complete the proof by establish the second part of the theorem.

Given $\mb{c}^S$, the expected payoff of a job from class $i$ with waiting cost $c \leq c_i^S$ from participating in the spot market is $v_i - \int_0^{c} w(t;\mb{c}^S)dt$. Since $\mb{c}^S$ satisfies \eqref{eq:spot-cutoffs-constraint}, the expected payoff is nonnegative. The pricing rule \eqref{eq:payment} ensures incentive compatibility for $c \leq c_i^S$. We only need to show that if $c > c_i^S$, the job does not participates in the spot market. Since $w(t;\mb{c}^S) = 1/\mu$ for $t \geq c_1^S$, reporting a waiting cost larger than $c_1^S$ does not improve the waiting time of the job, and the expected payment is at least $m(c_1^S)$. Hence, if a job with waiting cost $c > c_i^S$ decides to participate in the spot market, it will (mis)report a waiting cost $\wh{c} \in [0,c_1^S]$. The expected payoff of the job is $v_i - cw(\wh{c};\mb{c}^S) - m(\wh{c})$. Then,
\begin{align*}
& v_i - cw(\wh{c};\mb{c}^S) - m(\wh{c}) \\
& = v_i - c_i^Sw(\wh{c};\mb{c}^S) - m(\wh{c}) -(c-c_i^S)w(\wh{c};\mb{c}^S), \\
& \leq v_i - c_i^Sw(c_i^S;\mb{c}^S) - m(c_i^S) - (c-c_i^S)w(\wh{c};\mb{c}^S), \\
& = - (c-c_i^S)w(\wh{c};\mb{c}^S < 0.
\end{align*}
The first inequality is from the IC constraint \eqref{eq:ic} and then last equality is because the cutoffs $\mb{c}^S$ are the solutions of \eqref{eq:spot-cutoffs-eqn-repeat}. Hence the expected payoff of a job with waiting cost $c > c_i$ from participating in the spot market is negative and it will not participate. 

This completes the proof of Theorem \ref{thm:spot-equilibrium}. \hfill $\blacksquare$

\section{Proof of Theorem \lowercase{\ref{thm:hybrid-equilibrium}}} \label{sec:proof-hybrid-equilibrium}


\textbf{Step 1}: Existence of the unique solution $\overline{c}$

The existence and uniqueness of $\overline{c}_i$ follows from Lemma~\ref{lem:unique-solution}.

\textbf{Step 2}: Existence of unique $i$

In equilibrium, jobs of class $j$ whose waiting cost is $c_j(p)$ are indifferent between participating in the spot market and some outside option, either PAYG or balking.  That is,
$$v_j - \int_0^{c_j(p)} w(t; \mb{c}(p))dt = \max\left(0,v_j - \frac{p + c_j(p)}{\mu}\right),$$
or
$$\int_0^{c_j(p)} w(t; \mb{c}(p))dt = \min\left(v_j,\frac{p + c_j(p)}{\mu}\right).$$

The existence and uniqueness of such a $\mb{c}(p)$ follows from Lemma~\ref{lem:unique-solution}.  Given this solution, there is either a unique $c^*$ such that
$$\int_0^{c^*} w(t; \mb{c}(p))dt = \frac{p + c^*}{\mu},$$
or 
$$\int_0^{\mu v_1} w(t; \mb{c}(p))dt < \frac{p + \mu v_1}{\mu}.$$
In the latter case, no job participates in PAYG so we are done.  In the former, there is a unique $i$ such that
$$v_{i+1} \leq \frac{p + c^*}{\mu} < v_i,$$
or
$$p \in [\mu v_{i+1} - c^*, \mu v_i - c^*).$$
By construction, $c^* \geq \overline{c}_j$ for $j > i$ and $c^* \leq \overline{c}_j$ for $j \leq i$.  Thus there is a unique such $i$ as desired.

\textbf{Step 3}: Existence of the solution to the equations governing the choice of $\mb{c}(p)$.

Our argument in the previous step shows that the $\mb{c}(p)$ we construct is a solution to \eqref{hybrid-cutoff-eqn}.  By Lemma~\ref{lem:unique-solution} it is the unique solution.

\textbf{Step 4}: Characterizing equilibrium.

To be an equilibrium, the cutoff vector $\mb{c}(p)$ must satisfy the following constraints for all $c < c_j(p)$: 
\begin{equation}\label{eq:hcutoff-constraint1}
\begin{array}{l}
v_j - \int_0^{c}w(t; \mb{c}(p))dt \geq 0 \text{ and }\\
v_j - \int_0^{c}w(t; \mb{c}(p))dt > v_j- \frac{p+c}{\mu}.
\end{array}
\end{equation}
\begin{equation}\label{eq:hcutoff-constraint2}
\begin{array}{l}
\text{Hence either, } \left\{ 
\begin{array}{l}
   v_j - \int_0^{c_j(p)}w(t; \mb{c}(p))dt = 0, \\
   v_j - \frac{p+c}{\mu} < 0 \text{ for } c \in [c_j(p), \mu v_j - p],
\end{array} \right. \\
\text{or, } \left\{ 
\begin{array}{l}
   v_j - \int_0^{c_j(p)}w(t; \mb{c}(p))dt = v_j -  \frac{p+c_j(p)}{\mu}, \\
   v_j - \frac{p+c}{\mu} \geq v_j - \int_0^{c_i(p)}w(t; \mb{c}(p))dt \text{ and }\\
   v_j - \frac{p+c}{\mu} \geq 0 \text{ for } c \in [c_j(p), \mu v_j - p].
\end{array} \right.
\end{array}
\end{equation}
Constraint \eqref{eq:hcutoff-constraint1} says that the jobs with waiting cost below the cutoff get nonnegative expected payoff from participating in the spot market. Moreover, this expected payoff is strictly higher than that from participating in PAYG. Constraint \eqref{eq:hcutoff-constraint2} says that either no jobs from a class $i$ participate in PAYG, or jobs split between the spot market and PAYG with those above the cutoff weakly preferring PAYG.

First, consider classes $j \geq i$.  For these classes, $\int_0^{c_i(p)}w(t; \mb{c}(p))dt = v_i$, so by construction the first inequality of \eqref{eq:hcutoff-constraint1} is satisfied.  For the second, $\frac{p+c}{\mu} - \int_0^{c}w(t; \mb{c}(p))dt$ is decreasing in $c$ for $c < c_1(p)$, where it reaches 0.  Therefore the second inequality is satisfied.  For \eqref{eq:hcutoff-constraint2}, these classes satisfy the first system of constraints.  By construction the equality is satisfied, while the inequality is satisfied because $\frac{p+c}{\mu} - \int_0^{c}w(t; \mb{c}(p))dt > 0$.

Now consider classes $j < i$.  These classes share the same cutoff $c^*$ satisfying $\int_0^{c^*}w(t; \mb{c}(p))dt = \frac{p+c^*}{\mu}$.  For these classes, $v_j > \frac{p + c_j(p)}{\mu} > 0$, so \eqref{eq:hcutoff-constraint1} is satisfied.  For \eqref{eq:hcutoff-constraint2}, these classes satisfy the second system of constraints.  By construction the equality is satisfied, while the first inequality is satisfied because $\frac{p+c}{\mu} - \int_0^{c}w(t; \mb{c}(p))dt$ is decreasing in $c$ for $c < c_j(p)$, where it reaches 0.  The second inequality trivially holds for all $c \leq \mu v_i - p$. (This point is part of the intuition for our results about revenue).

This gives the desired equilibrium characterization. A job from class $j < i$ with waiting cost $c$ participates in the spot market if and only if $c \leq c_j(p)$; it participates in PAYG if $c_j(p) \leq c \leq \mu v_j - p$.  A job from class $j \geq i$ with waiting cost $c$ participates in the spot market if and only if $c \leq c_j(p)$; it never participates in PAYG.

\end{document}